\documentclass[USenglish,letterpaper,numberwithinsect,cleveref,autoref]{lipics-v2019}
\usepackage{algorithm}
\usepackage[noend]{algpseudocode}
\usepackage{xspace}
\usepackage{microtype}
\usepackage{xcolor}
\usepackage{sidecap}
\usepackage{thmtools, thm-restate}

\title{Optimal Joins using Compact Data Structures}

\author{Gonzalo Navarro}{
	DCC, U. of Chile \& IMFD, Chile}{}{}{}
  
  \author{Juan L. Reutter}{
	DCC, PUC \& IMFD, Chile}{}{}{}
  
  \author{Javiel Rojas-Ledesma}{
	DCC, U. of Chile \& IMFD, Chile}{}{}{}

\authorrunning{Navarro, G. \& Reutter, J. \& Rojas-Ledesma, J.}

\funding{This work was partially funded by Millennium Institute for Foundational Research on Data (IMFD), and by project CONICYT Fondecyt/Postdoctorado No. 63130209.}
\Copyright{Navarro, G. \& Reutter, J. \& Rojas-Ledesma, J.}



\nolinenumbers

\EventEditors{Carsten Lutz and Jean Christoph Jung}
\EventNoEds{2}
\EventLongTitle{23rd International Conference on Database Theory (ICDT 2020)}
\EventShortTitle{ICDT 2020}
\EventAcronym{ICDT}
\EventYear{2020}
\EventDate{March 30--April 2, 2020}
\EventLocation{Copenhagen, Denmark}
\EventLogo{}
\SeriesVolume{155}
\ArticleNo{24}

\theoremstyle{plain}
\newtheorem{observation}[theorem]{Observation}

\newcommand{\internal}{\textonehalf\xspace}

\newcommand{\Qdags}{\textrm{Qdags}\xspace}
\newcommand{\qdag}{\textrm{qdag}\xspace}
\newcommand{\qdags}{\textrm{qdags}\xspace}

\newcommand{\Qtree}{\textrm{Quadtree}\xspace}
\newcommand{\Qtrees}{\textrm{Quadtrees}\xspace}
\newcommand{\qtree}{\textrm{quadtree}\xspace}
\newcommand{\qtrees}{\textrm{quadtrees}\xspace}

\newcommand{\lqdag}{\textrm{lqdag}\xspace}
\newcommand{\lqdags}{\textrm{lqdags}\xspace}
\newcommand{\Lqdag}{\textrm{Lqdag}\xspace}
\newcommand{\Lqdags}{\textrm{Lqdags}\xspace}

\newcommand{\unknown}{$\Diamond$\xspace}

\newcommand{\qtreel}{\texttt{{QTREE}}\xspace}
\newcommand{\andl}{\texttt{{AND}}\xspace}
\newcommand{\orl}{\texttt{{OR}}\xspace}
\newcommand{\notl}{\texttt{{NOT}}\xspace}
\newcommand{\extendl}{\texttt{{EXTEND}}\xspace}
\newcommand{\joinl}{\texttt{{JOIN}}\xspace}
\newcommand{\diffl}{\texttt{{DIFF}}\xspace}

\newcommand{\extend}{\text{\sc Extend}\xspace}
\newcommand{\oror}{\text{{\sc Or}}\xspace}
\newcommand{\notnot}{\text{{\sc Not}}\xspace}

\newcommand{\andand}{\text{\sc And}\xspace}

\newcommand{\qtvalue}{\text{\sc Value}\xspace}
\newcommand{\childat}{\text{\sc ChildAt}\xspace}

\newcommand{\A}{\mathcal{A}}

\newcommand{\Dat}{\mathcal{D}}

\newcommand{\prefix}{\text{pre}}
\newcommand{\LongState}[1]{\State
	\parbox[t]{.95\linewidth}{#1\strut}}

\begin{document}
\maketitle
\sloppy

\begin{abstract}
Worst-case optimal join algorithms have gained a lot of attention in the database literature. We now count with several algorithms that are optimal in the worst case, and many of them have been implemented and validated in practice. However, the implementation of these algorithms often requires an enhanced indexing structure: to achieve optimality we either need to build completely new indexes, or we must populate the database with several instantiations of indexes such as B$+$-trees. Either way, this means spending an extra amount of storage space that may be non-negligible.

We show that optimal algorithms can be obtained directly from a representation that regards  the relations as point sets in variable-dimensional grids, without the need of  extra storage. Our representation is a compact \qtree for the static indexes, and a dynamic \qtree sharing subtrees (which we dub a \qdag) for intermediate  results.
We develop a compositional algorithm to process  full join queries under this representation, and show that the running time of this algorithm is worst-case optimal in data complexity. 
Remarkably, we can extend our framework to evaluate more expressive queries from relational algebra by introducing a lazy version of \qdags (\lqdags). Once again, we can show that the running time of our algorithms is worst-case optimal.
\end{abstract}

\maketitle

\section{Introduction} \label{sec-intro}

The state of the art in query processing has recently been shaken by a new generation of join algorithms with strong optimality guarantees based on the AGM bound 
of queries: the maximum size of the output of the query over all possible relations with the same cardinalities~\cite{AGM13}. One of the basic principles of these algorithms is to disregard the traditional notion of a query plan, favoring a strategy that can take further advantage of the structure of the query, while at the same time taking into account the actual size of the database~\cite{tutorialngo,skewstrikesback}.

Several of these algorithms have been implemented and tested in practice with positive results~\cite{HRRS19,NABKNRR15}, especially when handling queries 
with several joins. Because they differ from what is considered standard in relational database systems, the implementation of these algorithms often requires additional data structures, a database that is heavily indexed, or heuristics to compute the best computation path given the indexes that are present. For example, algorithms such as Leapfrog~\cite{leapfrog}, Minesweeper~\cite{minesweeper}, or InsideOut~\cite{insideout} must select a \emph{global order} on the attributes, and assume that relations are indexed in a way that is consistent with these attributes~\cite{NABKNRR15}. If one wants to use these algorithms with more flexibility in the way attributes are processed, then one would probably need to instantiate several combinations of B+ trees or other indexes~\cite{HRRS19}. On the other hand, more involved algorithms such as Tetris~\cite{geometric} or Panda~\cite{panda} require heavier data structures that allow reasoning over potential tuples in the answer.

Our goal is to develop optimal join algorithms that minimize the storage for additional indexes while at the same time being independent of a particular ordering of attributes. We address this issue by resorting to compact data structures: indexes using a nearly-optimal amount of space while supporting all operations we need to answer join queries. 

We show that worst-case optimal algorithms can be obtained when one assumes that the input data is represented as \qtrees, and stored under a compact representation for cardinal trees~\cite{BLNis13}.
\Qtrees are geometric structures used to represent data points in grids of size $\ell \times \ell$ (which can be generalized to any dimension). Thus, a relation $R$ with attributes $A_1,\dots,A_d$ can be naturally viewed as a set of points over grids of dimension $d$, one point per tuple of $R$: the value of each attribute $A_i$ is the $i$-th coordinate of the corresponding point.
%

To support queries under this representation, our main tool is a new dynamic version of \qtrees, which we denote \qdags, where some nodes may share complete subtrees.
Using \qdags, we can reduce the computation of a full join query $J = R_1 \bowtie \cdots \bowtie R_n$ with $d$ attributes, to an algorithm that first extends the \qtrees for $R_1,\dots,R_n$ into \qdags of dimension $d$, and then intersects them to obtain a \qtree. Our first result shows that such algorithm is indeed worst-case optimal: 
\begin{theorem}
Let $R_1(\A_1),\ldots,R_n(\A_n)$ be $n$ relations with attributes in $[0,\ell-1]$, and let $d = |\bigcup_i \A_i|$. We can represent the relations using only $\sum_i (|\A_i|+2+o(1)) |R_i| \log\ell + O(n\log d)$ bits, so that the result of a join query $J = R_1 \bowtie \cdots \bowtie R_n$ over a database instance $D$, can be computed in $\Tilde{O}(\text{AGM})$ time%
\footnote{
    $\Tilde{O}$ hides poly-log $N$ factors, for $N$ the total input size, as well as factors that just depend on $d$ and $n$ (i.e., the query size), which are assumed to be constant. We provide a precise bound in Section~\ref{sec:joinalg}.
}.
\end{theorem}
Note that just storing the tuples in any $R_i$ requires $|\A_i||R_i|\log\ell$ bits, thus our representation adds only a small extra space of $(2+o(1))|R_i|\log\ell + O(n\log d)$ bits (basically, two words per tuple, plus a negligible amount that only depends on the schema).  Instead, any classical index on the raw data (such as hash tables or B$+$-trees) would pose a linear extra space, $O(|\A_i||R_i|\log\ell)$ bits, often multiplied by a non-negligible constant (especially if one needs to store multiple indexes on the data).

Our join algorithm works in a rather different way than the most popular worst-case algorithms. To illustrate this, consider the triangle query $J = R(A,B) \bowtie S(B,C) \bowtie T(A,C)$.
The most common way of processing this query optimally is to follow what Ngo et al.~\cite{skewstrikesback} define as the \emph{generic} algorithm: select one of the attributes of the query (say $A$), and iterate over all elements $a \in A$ that could be an answer to this query, that is, all $a \in \pi_a(R) \cap \pi_a(T)$.
Then, for each of these elements, iterate over all $b \in B$ such that the tuple 
$(a,b)$ can be an answer: all $(a,b)$ in $(R \bowtie \pi_B(S)) \bowtie \pi_A(T)$, and so on. 

Instead, \qtrees divide the output space, which corresponds to a grid of size $\ell^3$, into 8 subgrids of size $(\ell/2)^3$, and for each of these grids it recursively evaluates the query.  As it turns out, this strategy is as good as the generic strategy defined by Ngo et al.~\cite{skewstrikesback} to compute joins,
and can even be extended to other relational operations, as we explain next. 

Our join algorithm boils down to two simple operations on \qtrees: an \extend operation that lifts the \qtree representation of a grid to a higher-dimensional grid, and an \andand operation that intersects trees. But there are other operations that we can define and implement. For example, the synchronized \oror of two \qtrees gives a compact representation of their union, and complementing the \qtree values can be done by a \notnot operation.
We integrate all these operations in a single framework, and use it to answer more complex queries given by the combination of these expressions, as in relational algebra. 

To support these operations we introduce lazy \qdags, or \lqdags for short, in which nodes may be additionally labeled with query expressions. The idea is to be able to delay the computation of an expression until we know such computation is needed to derive the output. To analyze our framework we extend the idea of a worst-case optimal algorithm to 
arbitrary queries: If a {\em worst-case optimal algorithm} to compute the output of a formula $F$ takes time $T$ over relations $R_1, \ldots, R_n$ of sizes $s_1, \ldots,
s_n$, respectively, of a database $D$, then there exists a database $D'$ with relations $R_1', \ldots, R_n'$ of sizes $O(s_1), \ldots, O(s_n)$,
respectively, where the output of $F$ over $R_1', \ldots, R_n'$ is of size $\Omega(T)$. We prove that \lqdags allow us to maintain optimality in the presence of union and negation operators: 

\begin{theorem}
Let $Q$ be a relational algebra query built with joins, union and complement, and where no relation appears more than once in $Q$. Then there is an algorithm to evaluate $Q$ that is worst-case optimal in data complexity. 
\end{theorem}

Consider, for example, the query $J' = R(A,B) \bowtie S(B,C) \bowtie \overline{T}(A,C)$, which joins $R$ and $S$ with the complement $\overline{T}$ of $T$. One could think of two ways to compute this query. The first is just to join $R$ and $S$ and then see which of the resulting tuples are not in $T$. But if $T$ is dense ($\overline{T}$ is small), it may be more efficient to first compute $\overline{T}$ and then proceed as on the usual triangle query. Our algorithm is optimal because it can choose locally between both strategies: by dividing into quadrants one finds dense regions of $T$ in which computing $\overline{T}$ is cheaper, while in sparse regions the algorithm first computes the join of $R$ and $S$. 

Our framework is the first in combining worst-case time optimality with the use of compact data structures. The latter can
lead to improved performance in practice, because relations can be stored in faster memory, higher in the memory hierarchy~\cite{Nav16}. This is especially relevant when the compact representation fits in main memory while a heavily indexed representation requires resorting to the disk, which is orders of magnitude slower. Under the recent trend of maintaining the database in the aggregate main memory of a distributed system, a compact representation leads to using fewer computers, thus reducing hardware, communication, and energy costs, while improving performance.




\section{Quadtrees} \label{sec-qtrees}

\begin{figure*}
	\centering
	\includegraphics[scale=0.7]{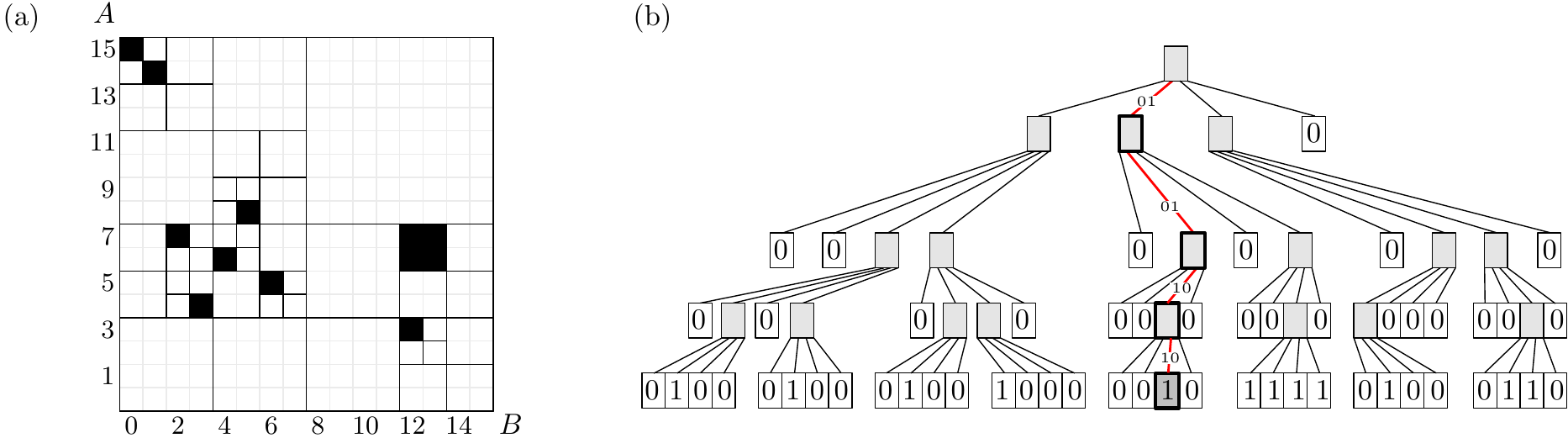}
	\caption{A \qtree representing $R(A, B) = \{(4,3), (7,2), (5,6), (6,4),
		(3,12), (6,12), (6,13)$, $(7,12), (7,13), (8,5), (14,1), (15,0)\}.$
		(a) Representation of $R(A,B)$ in a $2^4 \times 2^4$ grid, and
		representation of the hierarchical partition defining the \qtree. The
		black cells correspond to points in $R$.
		(b) The \qtree representing $R$. The shadowed leaf of the tree
		corresponds to the point $p = (3,12)$. Concatenating the labels in the
		path down to $p$ yield the bit-string
		`$0\textbf{1}0\textbf{1}1\textbf{0}1\textbf{0}$' which encodes the first
		(resp. second) coordinate of $p$ in the bits at odd (resp. even)
		positions ($3=0011, 12=1100$).
		%
		}
\label{fig:example}
\end{figure*}

A Region Quadtree~\cite{FB74,Sam06} is a structure used to store points in 
two-dimensional grids of $\ell \times \ell$. We focus on the variant called 
MX-Quadtree~\cite{WF90,Sam06}, which can be described as follows. Assume for
simplicity that $\ell$ is a power of 2. If $\ell=1$, then the grid has only one
cell and the \qtree is an integer $1$ (if the cell has a point) or $0$ (if
not). For $\ell>1$, if the grid has no points, then the 
\qtree is a leaf. Otherwise, the \qtree is an internal node with 
four children, each of which is the \qtree of one of the four $\ell/2 \times
\ell/2$ quadrants of the grid. (The deepest internal nodes, whose children
are $1 \times 1$ grids, store instead four integers in $\{ 0,1\}$ to encode 
their cells.)

Assume each data point is described using the binary representation of each of its coordinates (i.e., as a pair of $\log \ell$-bit vectors). 
We order the grid quadrants so that 
the first contains all points with coordinates of the form $(0 \cdot c_x, 0 \cdot c_y)$, for $\log \ell-1$ bit vectors $c_x$ and $c_y$, the second 
contains points $(0 \cdot c_x, 1 \cdot c_y)$, the third $(1 \cdot c_x, 0 \cdot
c_y)$, and the last quadrant stores the points $(1 \cdot c_x, 1 \cdot c_y)$. Fig.~\ref{fig:example} shows a grid and its deployment as a \qtree.

\Qtrees can be generalized to higher dimensions. A \qtree of dimension $d$ is a tree used to represent data points in a 
$d$-dimensional grid $G$ of size $\ell^d$. Here, an empty grid is represented by a leaf and a nonempty grid corresponds to an internal node with $2^d$ children representing the $2^d$ subspaces spanning from combining the first bits of each dimension. Generalizing the case $d=1$,
the children are ordered using the Morton~\cite{Mor66} partitioning of the grid: a sequence of $2^d$ subgrids of size $(\ell/2)^d$ in which the $i$-th 
subgrid of the partition, represented by the binary encoding $b_i$ of $i$, is
defined by all the points $(b_{c_1},\dots,b_{c_d})$ in which the word formed by 
concatenating the first bit of each string $b_{c_j}$ is precisely the string $b_i$. 

A \qtree with $p$ points has at most $p\log\ell$ nodes (i.e., root-to-leaf paths).
A refined analysis in two dimensions~\cite[Thm.~1]{GHKNPSdcc15.2} shows that \qtrees have
fewer nodes
when the points are clustered: if the points distribute along $c$ clusters,
$p_i$ of them inside a subgrid of size $\ell_i \times \ell_i$, then there are
in total $O(c\log \ell + \sum_i p_i \log \ell_i)$ nodes in the \qtree. The
result easily generalizes to $d$ dimensions: the cells are of size $\ell_i^d$
and the \qtree has $O(c\log \ell + \sum_i p_i \log \ell_i)$ internal nodes,
each of which stores $2^d$ pointers to children (or integers, in the last 
level).

Brisaboa et al.~\cite{BLNis13} introduced a compact \qtree representation called the $k^d$-tree.
They represent each internal \qtree node as the $2^d$ bits telling which of its 
quadrants is empty (0) or nonempty (1). Leaves and single-cell nodes are not 
represented because their data is deduced from the corresponding bit of their 
parent.
The $k^d$-tree is simply a bitvector $V$ obtained by concatenating the
$2^d$ bits of every (internal) node in levelwise order. Each node is identified with
its order in this deployment, the root being 1. Navigation on the \qtree (from a parent to its children, and vice-versa) is simulated in constant time using $o(|V|)$ additional bits on top of $V$.
On a \qtree in dimension $d$ storing $p$ points, the length of the bitvector $V$ is $|V| \le 2^d p \log \ell$, increasing exponentially with $d$. This bitvector is sparse, however, because it has at most $p\log \ell$ 1s, one per \qtree node. We then resort to a representation of high-arity cardinal trees introduced by Benoit et al.~\cite[Thm.~4.3]{BDMRRR05}, which requires only $(d+2)p\log\ell + o(p\log\ell) + O(\log d)$ bits, and performs the needed tree traversal operations in constant time.

\begin{observation}{(\textbf{cf.\ Benoit et al.~\cite{BDMRRR05}, Thm.~4.3})\hspace{4pt}}
\label{obs:compactqtree}
Let $Q$ be a \qtree storing $p$ points in $d$ dimensions with integer coordinates in the interval $[0, \log \ell - 1]$. Then,
there is a representation of $Q$ which uses $(d+2+o(1))p\log\ell + O(\log d)$ bits, can be constructed in linear expected time, and supports constant time
parent-children navigation on the tree.
\end{observation}

From now on, by \qtree we refer to this compact representation. Next, we show how to represent relations using \qtrees and evaluate join queries over this representation.


\section{Multi-way Joins using Qdags}\label{sec-join}

We assume for simplicity that the domain $\Dat(A)$ of an attribute $A$ consists of all binary strings of length $\log \ell$, representing the integers in $[0,\ell -1]$, and that $\ell$ is a power of $2$. 

A relation $R(\A)$ with attributes $\A=\{A_1,\dots,A_d\}$ can be naturally represented as a \qtree: simply interpret each tuple in $R(\A)$ as a data point over a $d$-dimensional grid with $\ell^d$ cells, and store those points in a $d$-dimensional \qtree. 
Thus, using \qtrees one can represent the relations in a database using compact
space. The convenience of this representation to handle restricted join queries
with naive algorithms has been demonstrated practically on RDF stores
\cite{AGBFMPNkais14}. In order to obtain a general
algorithm with provable performance,
we introduce \qdags, an enhanced version of \qtrees, together with a new algorithm to efficiently evaluate join queries over the compressed representations of the relations.

We start with an example to introduce the basics behind our algorithms and argue for the need of \qdags.
We then formally define \qdags and explore their relation with \qtrees.
Finally, we provide a complete description of the join algorithm and analyze its running time.

\subsection{The triangle query: quadtrees vs qdags}

Let $R(A,B)$, $S(B,C)$, $T(A,C)$ be relations over the attributes $\{A, B, C\}$ denote the domains of $A, B$ and $C$ respectively, and consider the triangle query $R(A,B) \bowtie S(B,C) \bowtie T(A,C)$.
The basic idea of the algorithm is as follows:
we first compute a \qtree $Q_R^*$ that represents the cross product 
$R(A,B) \times \text{All}(C)$, where $\text{All}(C)$ is a relation with an 
attribute $C$ storing all elements in the domain $[0,\ell-1]$. 
Likewise, we compute $Q_S^*$ representing $S(B,C) \times \text{All}(A)$,
and  $Q_T^*$ representing $T(A,C) \times \text{All}(B)$.
Note that these \qtrees represent points in the three-dimensional
 grid with a cell for every possible value in $\Dat(A) \times \Dat(B) \times \Dat(C)$, where we assume that the domains $\Dat(\cdot)$ of the attributes are
all $[0,\ell-1]$.
Finally, we traverse the three \qtrees in synchronization building 
a new \qtree that represents the intersection of $Q_R^*$, $Q_S^*$ and $Q_T^*$. 
This \qtree represents the desired output because 
%
$$	R(A,B) \bowtie S(B,C) \bowtie T(A,C) = 
		(R(A,B) \times \text{All}(C)) \cap (S(B,C) \times \text{All}(A)) \cap (T(A,C) \times \text{All}(B)).$$

Though this algorithm is correct, it can perform poorly in terms of space and running time. The size of $Q_R^*$, for instance, can be considerably bigger than that of $R$, and even than the size of the output of the query. 
If, for example, the three relations have $n$ elements each, the size of the output is bounded by $n^{3/2}$ \cite{AGM13}, while building $Q_R^*$ costs $\Omega(n \ell)$ time and space.
This inefficiency stems from the fact that \qtrees are not smart to represent relations of the form $R^*(\A)=R(\A')\times \text{All}(\A \setminus \A')$, where $\A' \subset \A$, with respect to the size of a \qtree representing $R(\A')$.
Due to its tree nature, a \qtree does not benefit from the regularities that appear in the grid representing $R^*(\A)$. 
To remedy this shortcoming, we introduce \qdags, 
\qtree-based data structures
that represent sets of the form $R(\A')\times \text{All}(\A \setminus \A')$ by adding only {\em constant} additional space to the \qtree representing $R(\A')$, for any $\A' \subseteq \A$.

A \qdag is an {\em implicit} representation of a $d$-dimensional
\qtree $Q^d$ (that has certain regularities) using only a reference to a $d'$-dimensional
\qtree $Q^{d'}$, with $d' \le d$, and an auxiliary {\em mapping function} that defines how to use $Q^{d'}$ to simulate navigation over $Q^{d}$.  
\Qdags can then represent relations of the form  $R(\A')\times \text{All}(\A \setminus \A')$
using only a reference to a \qtree 
representing $R(\A')$, and a constant-space mapping function.

To illustrate how a \qdag works, consider a relation $S(B,C)$, and let $Q_S^*$ be 
a \qtree representing $S^*(A,B,C) = \text{All}(A) \times S(B,C)$.
Since $Q_S^*$ stores points in the $\ell^3$ cube, each node in $Q_S^*$ 
has $8$ children.
As $\text{All}(A)$ contains all $\ell$ elements, for each 
original point $(b,c)$ in $S$, $S^*$ contains $\ell$ points corresponding to elements
$(0,b,c),\dots,(\ell-1,b,c)$.
We can think of this as \emph{extending} each point in $S$ to a box of dimension $\ell
\times 1 \times 1$.
With respect to $Q_S^*$, this implies that, among the $8$ children of a node, 
the last $4$ children will always be identical to the first $4$, and their values
will in turn be identical to those of the corresponding nodes in the \qtree $Q_S$
representing $S$.
In other words, each of the 
four subgrids $1 a_1 a_2$ is identical to the subgrid $0 a_1 a_2$, and these 
in turn are identical to the subgrid $a_1 a_2$ in $S$
(see Fig.~\ref{fig:example2} for an example).
Thus, we can implicitly represent $Q_S^*$ by the pair $(Q_S, M=[0,1,2,3,0,1,2,3])$:
the root of $Q_S^*$ is the root of $Q_S$,
and the $i$-child of the root of $Q_S^*$ is represented by the pair $(C, M)$, 
where $C$ is the $M[i]$-th child of the root of $Q_S$.
%

\begin{figure*}[t]
	\centering
	\includegraphics[width=\textwidth,page=3]{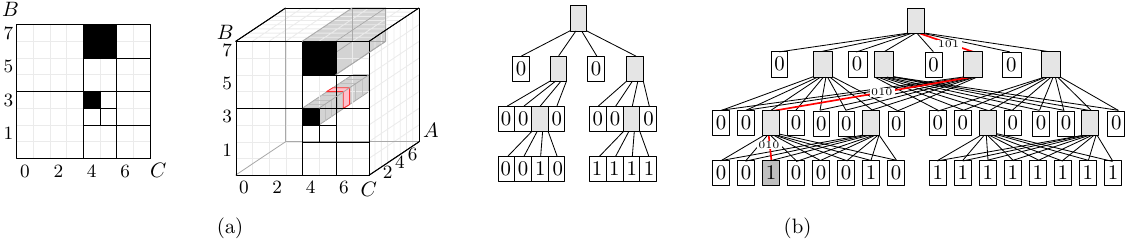}
	\vspace*{-8mm}
	\caption{An illustration of a \qdag for $S^*(\{A,B,C\})=\text{\normalfont
		All}(A) \times S(B,C)$, with $S(B, C) = \{(3,4), (6,4), (6,5), (7,4),
		(7,5)\}$. 
		a) A geometric representation of $S(B, C)$ (left), and
		$S^*(\{A,B,C\})$ (right). 
		b) A \qtree $Q_S$ for $S(B, C)$ (left), and the directed acyclic graph
		induced by the \qdag $(Q_S, M=[0,1,2,3,0,1,2,3])$, which represents
		$S^*(\{A,B,C\})$. The red cell in (a) corresponds to the point
		$p=(4,3,4)$. The leaf representing $p$ in the \qdag can be reached
		following the path highlighted in (b). Note the relation between the
		binary representation ({\normalfont \textbf{1}{\color{gray}
		\textbf{0}}0,\textbf{0}{\color{gray}
		\textbf{1}}0,\textbf{1}{\color{gray} \textbf{0}}0}) of $p$, and the
		Morton codes \textbf{101}, {\color{gray} \textbf{010}}, {\normalfont
		010} of the nodes in the path from the root to the leaf for $p$.}
	\label{fig:example2}
\end{figure*}

\subsection{Qdags for relational data}

We now introduce a formal definition of the \qdags, and describe the algorithms
which allow the evaluation of multijoin queries in worst-case optimal time.

\begin{definition}[\qdag] \label{def:qdag}
	Let $Q^{d'}$ be a \qtree representing a relation with $d'$ attributes.
	A \qdag $Q^d$, for $d \ge d'$, is a pair $(Q^{d'},M)$, with 
$M:[0,2^{d}-1]\rightarrow[0,2^{d'}-1]$. This $\qdag$ represents a $\qtree$ $Q$,
which is called the {\em completion} of $Q^d$, as follows:
	\begin{enumerate}
		\item If $Q^{d'}$ represents a single cell, then $Q$ 
represents a single cell with the same value. 
		\item  If $Q^{d'}$ represents a $d'$-dimensional grid empty of
points, then $Q$ represents a $d$-dimensional grid empty of points.
		\item Otherwise, the roots of both $Q^{d'}$ and $Q$ are 
internal nodes, and for all $0 \le i < 2^d$, the $i$-th child of $Q$ is the
\qtree represented by the \qdag $(C(Q^{d'},M[i]),M)$, where $C(Q^{d'},j)$ 
denotes the $j$-th child of the root node of \qtree $Q^{d'}$.
	\end{enumerate}	
\end{definition}

We say that a \qdag represents the same relation $R$ represented by its 
completion. 
Note that, for any $d$-dimensional \qtree $Q$, one can generate a \qdag whose completion is $Q$ simply as the pair $(Q, M)$, where $M$ is the \emph{identity mapping} $M[i]=i$, for all $0 \le i < 2^d$.
We can then describe all our operations over \qdags. Note, in particular, that we can use mappings to represent any reordering of the attributes.

In terms of representation, the references to \qtree nodes consist of the
identifier of the \qtree and the index of the node in level-wise order. This
suffices to access the node in constant time from its 
compact representation.
For a \qdag $Q' = (Q,M)$, we denote by $|Q'|$ the number of internal nodes in 
the base \qtree $Q$, and by $||Q'||$ the number of internal nodes in the
completion of $Q'$.

\begin{figure*}[t]
	\begin{minipage}[t]{0.49\textwidth}
		\begin{algorithm}[H]
			\begin{algorithmic}[1]
				\footnotesize
				\Require \qdag $(Q,M)$ with grid side $\ell$.
				\Ensure The integer $1$ if the grid is a single point, $0$ if the grid is empty, and \internal otherwise.
				\vskip 4pt
				\If {$\ell = 1$} \Return the integer $Q$ \EndIf
				\If {$Q$ is a leaf} \Return $0$ \EndIf
			\State \Return \internal
			\end{algorithmic}
		\caption{\qtvalue}\label{alg:value}
		\end{algorithm}
	\end{minipage}
	\hskip .03\textwidth
	\begin{minipage}[t]{0.49\textwidth}
	\small
	\begin{algorithm}[H]
			\begin{algorithmic}[1]
				\footnotesize
	\Require \qdag $(Q,M)$ on a grid of dimension $d$ and side
	$\ell$, and a child number $0 \le i < 2^d$. Assumes $Q$ is not a leaf or an integer.
					\Ensure A \qdag $(Q',M)$ corresponding to the $i$-th child of
	$(Q,M)$.
	
					\vskip 4pt
					\State{\Return $(C(Q,M[i]),M)$}
			\end{algorithmic}
			\caption{\childat}
			\label{alg:child}
	\end{algorithm}
	\end{minipage}	
\end{figure*}

Algorithms~\ref{alg:value} and \ref{alg:child}, based on Definition~\ref{def:qdag},
will be useful for the navigation of \qdags. Operation \qtvalue yields a 0 iff
the subgrid represented by the \qdag is empty (thus the \qdag is a leaf); a $1$
if the \qdag is a full single cell, and \internal if it is an internal node. Operation \childat lets us
descend by a given child from internal nodes representing nonempty grids. The
operations ``integer $Q$'', ``$Q$ is a leaf'', and ``$C(Q,j)$'' are
implemented in constant time on the compact representation of $Q$.

\smallskip
\noindent
\textbf{Operation \extend}. We introduce an operation to obtain, from the \qdag
representing a relation $R$, a new \qdag representing the relation $R$ extended
with new attributes.

\begin{definition} \label{def:extend}
	Let $\A' \subseteq \A$ be sets of attributes, let $R(\A')$ be a relation
over $\A'$, and let $Q_R=(Q,M)$ be a \qdag that represents $R(\A')$. The
operation $\extend(Q_R,\A)$ returns a \qdag $Q_R^*=(Q,M')$ that represents the
relation $R \times \text{All}(\A \setminus \A')$.
\end{definition}

To provide intuition on its implementation, let $\A'$ be the set of attributes
$\{A,B,D\}$ and let $\A = \{A,B,C,D\}$, and consider $R(\A')$, $Q_R$ and $Q_R^*$
from Definition~\ref{def:extend}. Each node of $Q_R$ has 8 children, while each
node of $Q_R^*$ has 16 children. Consider the child at position $i=12$ of
$Q_R^*$. This node represents the grid with Morton code
$m_4$=`\textbf{11{\color{gray}0}0}' (i.e., 12 in binary), and contains the
tuples whose coordinates in binary start with \textbf{1} in attributes $A,B$ and
with \textbf{0} in attributes $C, D$. This  child has elements if and only if
the child with Morton code $m_3$=`\textbf{110}' of $Q_R$ (i.e., its child at
position $j=6$) has elements; this child is in turn the $M[6]$-th child of $Q$.
Note that $m_3$ results from projecting $m_4$ to the positions 0,1,3 in which the attributes $A, B, D$ appear in $\{A,B,C,D\}$. 
Since the Morton code \textbf{11{\color{gray}1}0}' (i.e., 14 in binary) also projects to $m_3$, it holds that $M'[12]=M'[14]=M[6]$.
We provide an implementation of the \extend{} operation for the general case in Algorithm~\ref{alg:extend-body}. 
The following lemma states
the time and space complexity of our implementation of \extend{}. For simplicity, we count the space in terms of computer
words used to store references to the \qtrees and values of the mapping function
$M$.

\begin{algorithm}[t]
	\footnotesize
	\begin{algorithmic}[1]
		\footnotesize
		\Require A \qdag $(Q,M')$ representing a relation $R(\A')$,
and a set $\A$ such that $\A' \subseteq \A$.
		\Ensure A \qdag $(Q,M)$ whose completion represents the relation $R(\A')\times \text{All}(\A \setminus \A')$.
		
		\State create array $M[0,2^d-1]$
		\State $d \gets |\A|$, $d' \gets |\A'|$ 
		\For{$i \leftarrow 0,\ldots, 2^d-1$}
		\State $m_d \leftarrow $ the $d$-bits binary representation of $i$
		\LongState{$m_{d'} \leftarrow$ the projection of $m_d$ to the positions in which the attributes of $\A'$ appear in $\A$}
		\LongState{$i' \leftarrow $ the value in $[0,2^{d'}-1]$ corresponding to $m_{d'}$}
		\State $M[i] \gets M'[i']$
		\EndFor			
		\State \Return $(Q, M)$		
	\end{algorithmic}
	\caption{\texttt{\extend}}	
	\label{alg:extend-body}
\end{algorithm}

\begin{restatable}{lemma}{extendlemma}
 \label{lem:extend}
Let $|\A|=d$ in Definition~\ref{def:extend}. Then, the operation
$\extend(Q_R,\A)$ can be supported in time $O(2^d)$ 
and its output takes $O(2^d)$ words of space.
\end{restatable}
\begin{proof}
We show that Algorithm~\ref{alg:extend-body} meets the conditions of the lemma.
The computations of $m_d$ 
and $i'$ are immaterial (they just interpret a bitvector as a number or vice 
versa). The computation of $m_d'$ is done with a constant table (that
depends only on the database dimension $d$) of size $O(2^{3d})$:%
\footnote{They can be reduced to two tables of size $O(2^{2d})$, but we omit
the details for simplicity.}
The argument $\A$ is given as a bitvector of size $d$ telling which
attributes are in $\A$, the \qdag on $\A'$ stores a bitvector of size $d$ 
telling which attributes are in $\A'$, and the table receives both bitvectors
and $m_d$ and returns $m_d'$.
\end{proof}


\subsection{Join algorithm} \label{sec:joinalg}

Now that we can efficiently represent relations of the form $R(\A')\times \text{All}(\A \setminus \A')$, for $\A' \subseteq \A$, we describe a worst-case optimal implementation of joins over the \qdag representations of the relations.
Our algorithm follows the idea discussed for 
the triangle query: we first extend every \qdag to all the attributes that appear
in the query, so that they all have the same dimension and attributes.
Then we compute their intersection, building a \qtree representing the output of
the query. 
The implementation of this algorithm is surprisingly simple
(see Algorithms~\ref{alg:join} and \ref{alg:and-body}), 
yet worst-case optimal, as we prove later on.
Using \qdags is key for this result; this algorithm would not be at all
optimal if computed over relational instances stored using standard
representations such as B+ trees.
First, we describe how to compute the intersection of several \qdags, and then analyze the running time of the join.

\begin{figure*}[t]
    \begin{minipage}{.37\textwidth}
        \begin{algorithm}[H]
        	\footnotesize		
        	\caption{\small \textsc{MultiJoin}}
        	\label{alg:join}
        	\begin{algorithmic}[1]
        		\Require Relations $R_1,\dots,R_n$, stored as \qdags $Q_1,\dots,Q_n$;
        		each relation $R_i$ is over attributes $\A_i$ and $\A =
        \bigcup \A_i$. 
                \vskip 4pt
                \Ensure A quadtree representing the output of $J = R_1 \bowtie \dots \bowtie R_n$.
                
        		\vskip 10pt
        		\For{$i \leftarrow 1,\ldots,n$} 
        		\State $Q_i^* \leftarrow \extend(R_i,\A)$ \label{line:extend}
        		\EndFor
        		
        		\vskip 10pt
        		\State \Return $\andand(Q_1^*,\dots,Q_n^*)$ 
        	\end{algorithmic}
        \end{algorithm}
    \end{minipage}
    \hspace{.03\textwidth}
    \begin{minipage}{.6\textwidth}
        \begin{algorithm}[H]
        	\begin{algorithmic}[1]
        		\footnotesize		
        		\Require $n$ \qdags $Q_1$, $Q_2$, \ldots ,$Q_n$ representing
        relations $R_1(\A), R_2(\A), \ldots, R_n(\A)$.
        		\Ensure A \qtree representing the relation $\bigcap_{i=1}^{n}
        R_i(\A)$.
        		
        		\vskip 4pt
        		\State $m \leftarrow \min \{\qtvalue(Q_1), \ldots,
        \qtvalue(Q_n)\}$
        		\If {$\ell = 1$} \Return the integer $m$ \EndIf
        		\If {$m=0$} \Return a leaf \EndIf
        		\For {$i \leftarrow 0,\ldots,2^d-1$}
        			\State $C_i \gets \andand(\childat(Q_1,i),\ldots,$ $\childat(Q_n, i))$
        		\EndFor
        		\If {$\max \{\qtvalue(C_0), \ldots,
        \qtvalue(C_{2^{d}-1})\} = 0$}
        		\Return a leaf \EndIf
        		\State \Return a \qtree with children $C_0, \ldots, C_{2^d-1}$
        	\end{algorithmic}
        	\caption{\texttt{\andand}}
        	\label{alg:and-body}	
        \end{algorithm}
    \end{minipage}
\end{figure*}

\smallskip
\noindent
\textbf{Operation \andand}. We introduce an operation \andand, which computes the
intersection of several relations represented as \qdags.

\begin{definition} \label{def:and}
Let $Q_1,\dots,Q_n$ be \qdags representing relations $R_1,\ldots,R_n$, all
over the attribute set $\A$. Operation 
$\andand(Q_1,\dots,Q_n)$ returns a \qtree
$Q$ that represents the relation $R_1 \cap \ldots \cap R_n$.
\end{definition}

We implement this operation by simulating a synchronized traversal among the
completions $C_1, \ldots, C_n$ of $Q_1,\ldots,Q_n$, respectively, obtaining the
\qtree $Q$ that stores the cells that are present in all the \qtrees $C_i$
(see  Algorithm~\ref{alg:and-body}).
We proceed as follows. If $\ell=1$, then all $C_i$ are integers with values $0$
or $1$, and $Q$ is an integer equal to the minimum of the $n$ values. Otherwise,
if any $Q_i$ represents an empty subgrid, then $Q$ is also a leaf representing
an empty subgrid. Otherwise, every $C_i$ is rooted by a node $v_i$ with $2^d$
children, and so is $Q$, where the $j$-th child of its root $v$ is the result of
the \andand operation of the $j$-th children of the nodes $v_1,\dots,v_n$.
However, we need a final {\em pruning} step to restore the \qtree invariants
(line 6 of Algorithm \ref{alg:and-body}): if $\qtvalue(v_i)=0$ for
all the resulting children of $v$, then $v$ must become a leaf and the children be
discarded.
Note that once the \qtree is computed, we can represent it succinctly in linear expected time so that, for instance, it can be cached for future queries involving the output represented by $Q$%
\footnote{
    This consumes linear expected time due to the use of perfect hashing in the compact representation~\cite{BDMRRR05}.
}.

\smallskip
\noindent
\textbf{Analysis of the algorithm}. We compute the output $Q$ of
$\andand(Q_1,\dots,Q_n)$ in time $O(2^d \cdot (||Q_1||+\cdots+||Q_n||))$. More
precisely, the time is bounded by $O(2^d \cdot |Q^+|)$, where $Q^+$ is the
\qtree that would result from Algorithm~\ref{alg:and-body} if we remove the pruning step of line 6. We name this \qtree $Q^+$ as the \emph{non-pruned version} of $Q$. Although the size of the actual output $Q$ can be much smaller than that of $Q^+$,we can still prove that our time is optimal in the worst case.
We start with a technical result.

\begin{restatable}{lemma}{andlemma}
	\label{lem:and}
The \andand operation can be supported in time $O(M \cdot 2^d n \log \ell)$, where $M$ is the maximum
number of nodes in a level of $Q^+$.
\end{restatable}
\begin{proof}
We show that Algorithm~\ref{alg:and-body} meets the conditions of the lemma.
Let $m_j$ be the number of nodes of depth $j$ in $Q^+$, and then
$M = \max_{0 \leq j < \log\ell} m_j$. The number of steps performed by 
Algorithm~\ref{alg:and-body} is bounded 
by $n \cdot (\sum_{0 \le j < \log \ell} m_j \cdot 2^d) \leq n \cdot M \cdot
\log \ell \cdot 2^d$: In each depth we continue traversing all \qdags $Q_1,\dots,Q_n$ as long as they are all nonempty, and we generate the corresponding nodes in $Q^+$ (even if at the end some nodes will disappear in $Q$).
\end{proof}

All we need to prove (data) optimality is to show that $|Q^+|$ is bounded by the size of the real output of the query.
Recall that, for a join query $J$ on a database $D$, we use $2^{\rho^*(J,D)}$ to denote the AGM bound~\cite{AGM13} of the query $J$ over $D$, that is, the maximum size of the output of $J$ over any relational database having the same number of tuples as $D$ in each relation.

\begin{theorem} \label{thm:main}
Let $J = R_1 \bowtie \dots \bowtie R_n$ be a full join query, and $D$ a database over schema $\{R_1,\dots,R_n\}$, with $d$ attributes in total, and where the 
domains of the relations are in $[0,\ell-1]$.  Let $\A_i$ be the set of attributes of $R_i$, for all $1 \le i \le n$ and $N=\sum_i{|R_i|}$ be the total amount of tuples in the database.
The relations $R_1, \ldots, R_n$ can then be stored within
$\sum_i (|\A_i|+2+o(1)) |R_i| \log\ell + O(n\log d)$
bits, so that the output for $J$ can be computed in time $O(2^{\rho^*(J,D)} \cdot 2^d n \log \min(\ell,N)) = \Tilde{O}(2^{\rho^*(J,D)})$.
\end{theorem} 
\begin{proof}
The space usage is a simple consequence of Observation~\ref{obs:compactqtree}. 
As for the time, to solve the join query $J$ we simply encapsulate the \qtrees representing $R_1, \ldots, R_n$ in \qdags $Q_1,\ldots,Q_n$, and use Algorithm~\ref{alg:join} to compute the result of the query.
We now show that Algorithm~\ref{alg:join} runs in time within the bound of the theorem.
First, assume that $\log \ell$ is $O(\log N)$. 
Let each relation $R_i$ be over attributes $\A_i$, and $\A = \bigcup \A_i$ with
$d = |\A|$. Let $Q_i^* = \extend(Q_i,\A)$, $Q = \andand(Q_1^*,\dots,Q_n^*)$,
and $Q^+$ be the non-pruned version of $Q$.
The cost of the \extend operations is only $O(2^d n)$, according to
Lemma~\ref{lem:extend}, so the main cost owes to the \andand operation.

If the maximum $M$ of Lemma~\ref{lem:and} is reached at the lowest level of 
the decomposition, where we store integers $0$ or $1$, then we are done: 
each $1$ at a leaf of $Q^+$ exists in $Q$ as well
because that single tuple is present in all the relations $R_1,\ldots,R_n$.
Therefore, $M$ is bounded by 
the AGM bound of $J$ and the time of the \andand\ operation is 
bounded by $O(2^{\rho^*(J,D)} \cdot 2^d n \log \ell)$. 

Assume instead that $M$ is the number of internal nodes at depth $0 < j < 
\log \ell$ of $Q^+$ (if $M$ is reached at depth $0$ then $M = 1$). 
Intuitively, we will take the relations at the granularity of level $j$, and
show that there exists a database $D'$ where such a $(2^j)^d$ relation
arises in the last level and thus the answer has those $M$ tuples.

We then construct the following database $D'$ with relations $R_i'$:
For a binary string $c$, let $\prefix(c,j)$ denote the first $j$ bits of $c$. Then, 
for each relation $R_i$ and each tuple $(c_1,\dots,c_{d_i})$ in $R_i$, where $d_i = |\A_i|$, 
let $R_i'$ contain the tuples $(0^{\log \ell - j}\prefix(c_1,j),0^{\log \ell - j}\prefix(c_2,j)\dots,$\linebreak $0^{\log \ell - j}\prefix(c_{d_i},j))$, corresponding to taking the first 
$j$ bits of each coordinate and prepending them with a string of $\log \ell -j$ $0$s. 
While this operation may send two tuples in a relation in $D$ to a single tuple in $D'$, we still have that each relation $R_i'$ in $D'$ contains at most as many 
tuples as relation $R_i$ in $D$. Moreover, if we again store every $R_i'$ as a
\qdag and process their join as in Algorithm~\ref{alg:join}, 
then by construction we have in this case that the leaves of the tree
resulting from the \andand\ operation contain exactly $M$ nodes with $1$, and
that this is the maximum number of nodes in a level of
this tree. Since the leaves represent tuples in the answer, we have that $M \leq 2^{\rho^*(J,D')} \leq 2^{\rho^*(J,D)}$, which completes the proof for the case when $\log \ell$ is $O(\log N)$.

Finally, when $\log N$ is  $o(\log \ell)$, we can convert $O(\log\ell)$ to $O(\log N)$ in the time complexities by storing $R_1,\dots,R_n$ using \qtrees, with a slight variation.
We store the values of the attributes appearing in any relation in an auxiliary data structure (e.g., an array), and associate an $O(\log N)$-bits identifier to each different value in $[0,\ell-1]$ that appears in $D$ (e.g., the index of the corresponding value in the array). In this case, we represent the relations in \qtrees, but using the identifiers of the attribute values instead of the values themselves. This representation requires at most $dN \log \ell$ bits for the representation of the distinct attribute values and $O(dN \log N)$ bits for the representation of the \qtrees. Thus, in total it requires $dN \log \ell + O(dN \log N) = dN \log \ell (1 + O(\log N / \log \ell)) = dN \log \ell (1 + o(1))$, which is within the stated bound.
Note that in both cases, the height of the \qtrees representing the relations is $O(\log N)$, and this is the multiplying factor in the time complexities.
\end{proof}

\section{Extending Worst-Case Optimality to More General Queries}\label{sec-workflow}


In this section we turn to design worst-case optimal algorithms for more expressive queries. At this point it should be clear that we can deal with set operations: we already studied the \emph{intersection} 
(which corresponds to operation \andand over the \qdags), and will show that 
\emph{union} (operation \oror) and \emph{complement} (operation
\notnot) can be solved optimally as well. 
What is most intriguing, however, is whether we can obtain worst-case
optimality on combined relational formulas.
We introduce a worst-case optimal algorithm to evaluate
formulas expressed as combinations of join, union, and complement operations (which we refer to as JUC-queries; note that intersection is a particular case of join). 
We do not study other operations like \emph{selection} and
\emph{projection} because these are easily solved
in time essentially proportional to the size of the output, but refer to Appendix~\ref{sec-more-stuff} for more details on how projection interplays with the rest of our framework. 

The key ingredient of our algorithm is to deal with these operations in a
\emph{lazy} form, allowing unknown intermediate results so that all components of a formula are evaluated simultaneously. To do this we introduce lazy \qdags (or \lqdags), an alternative to \qdags that can navigate over the \qtree representing the output of a formula without the need to entirely evaluate the formula.
We then give a worst-case optimal algorithm to compute the \emph{completion} of an \lqdag, that is, the \qtree of the grid represented by the \lqdag.

\subsection{Lqdags for relational formulas}

To support worst-case optimal evaluation of relational formulas we introduce
two new ideas: we add ``full leaves'' to the \qtree representation to denote
subgrids full of 1s, and we introduce  \lqdags to represent the result of a
formula as an {\em implicit} \qtree that can be navigated without fully evaluating the formula.


\begin{algorithm}[t]
\footnotesize
        \begin{algorithmic}[1]
                \Require \qdag $(Q,M)$ with grid side $\ell$.
                \Ensure Value $0$ or $1$ if the grid represented by $Q$ 
is totally empty or full, respectively, otherwise \internal.
                \vskip 4pt
                \If {$Q$ is a leaf} \Return the integer $0$ or $1$ associated with $Q$ \EndIf
                \State \Return \internal
        \end{algorithmic}
        \caption{\qtvalue on extended \qdags}
        \label{alg:value2}
\end{algorithm}

While \qtree leaves representing a single cell store the cell value, $0$ or 
$1$, \qtree leaves at higher levels always represent subgrids full of $0$s.
We now generalize the representation, so that \qtree leaves at any level store
an integer, $0$ or $1$, which is the value of all the cells in the subgrid 
represented by the leaf. The generalization impacts on the way to compute 
\qtvalue, depicted in Algorithm~\ref{alg:value2}. We will not use \qdags in
this section, however; the \lqdags build directly on \qtrees.
In terms of the compact 
representation, this generalization is 
implemented by resorting to an impossible \qtree configuration: an internal
node with all zero children \cite{dBABNP13}. Note that replacing a full subgrid
with this configuration can only decrease the size of the 
representation.

The second novelty, the \lqdags, are defined as follows. 

\begin{definition}[\lqdag] \label{def:lqdag}
	An \lqdag $L$ is a pair $(f, o)$, where $f$ is a \emph{functor} and $o$ is a list of \emph{operands}.
	The {\em completion} of $L$ is the \qtree $Q_R = Q_R(\A)$ 
representing relation $R(\A)$ if $L$ is as follows:
	\begin{enumerate}
		\item $(\qtreel, Q_R)$, where the \lqdag just represents $Q_R$;
		
		\item $(\notl, Q_{\overline{R}})$, where $Q_{\overline{R}}$ is
the \qtree representing the complement of $Q_R$;
		
		\item $(\andl, L_1, L_2)$, where $L_1$ and $L_2$ are \lqdags
and $Q_R$ represents the intersection of their completions;
		
		\item $(\orl, L_1, L_2)$, where $L_1$ and $L_2$ are \lqdags 
and $Q_R$ represents the union of their completions;
		
		\item $(\extendl, L_1, \A)$, where \lqdag $L_1$
represents $R'(\A')$, $\A' \subseteq \A$, and $Q_R$ represents $R(\A) = R'(\A')\times \text{All}(\A \setminus \A')$.
	\end{enumerate}
\end{definition}

Note that, for a \qtree $Q_R$ representing a relation $R(\A')$, and a set of attributes $\A$, the \qdag $Q_R^*=(Q_R,M_\A)$ that represents the relation $R \times \text{All}(\A \setminus \A')$ can be expressed as the \lqdag $(\extendl, (\qtreel,Q_R), \A)$. In this sense, \lqdags are extensions of \qdags.
To further illustrate the definition of \lqdags, consider the triangle query $R(A,B) \bowtie
S(B,C) \bowtie T(A,C)$, with $\A = \{A, B, C\}$ and the relations represented
by \qtrees $Q_R$, $Q_S$, and $Q_T$.
This query can then be represented as the \lqdag
\begin{eqnarray*}
(\andl,(\andl,(\extendl,(\qtreel,Q_R),\A),(\extendl,(\qtreel,Q_S),\A)),
(\extendl,(\qtreel,Q_T),\A)).
\end{eqnarray*}
It is apparent that one can define other operations, like \joinl and \diffl, by combining the operations defined above:
\begin{eqnarray*}
	(\joinl, L_1(\A_1), L_2(\A_2)) &=& (\andl, (\extendl, L_1, \A_1 \cup \A_2), (\extendl, L_2, \A_1 \cup \A_2)) \\
		(\diffl, L_1, L_2) &=& (\andl, L_1, (\notl, L_2))	
\end{eqnarray*}

Note that in the definition of the \lqdag for \notl, the operand is a 
\qtree instead of an \lqdag, and then, for example, $L_2$ should be a
\qtree in the definition of \diffl, in principle. However, we can easily get around that restriction by pushing down the \notl operators until the operand is a \qtree or the \notl is cancelled with
another \notl.
For instance, a \notl over an \lqdag $(\andl, Q_1, Q_2)$ is equivalent to 
$(\orl$, $(\notl, Q_1)$, $(\notl, Q_2))$, and analogously with the other 
functors.
The restriction, however, does limit the types of formulas for which we achieve
worst-case optimality, as shown later.

To understand why we call \lqdags lazy, consider the operation $Q_1\ \andand\ Q_2$ over \qtrees $Q_1, Q_2$.
If any of the values at the roots of $Q_1$ or $Q_2$ is $0$, then the result of the operation is for sure a leaf with value $0$.
If any of the values is $1$, then the result of the operation is the other.
However, if both values are \internal, one cannot be sure of the value of the
root until the \andand between the children of $Q_1$ and $Q_2$ has been evaluated.
Solving this dependency eagerly would go against worst-case optimality: it
forces us to fully evaluate parts of the formula without considering it as a whole.
To avoid this, we allow the \qtvalue of a node represented by an \lqdag to be, apart from $0$, $1$, and \internal, the special value \unknown. 
This indicates that one cannot determine the value of the node without computing the values of its children. 

As we did for \qdags, in order to simulate the navigation over the completion $Q$ of an \lqdag $L$ we need to describe how to obtain the value of the root of $Q$, and how to obtain an \lqdag whose completion is the $i$-th child of $Q$, for any given $i$.
We implement those operations in
Algorithms~\ref{alg:valuenot-body}--\ref{alg:childatextend-body}, all constant-time.
Note that \childat can only be invoked when $\qtvalue=$ \internal or \unknown.
The base case is $\qtvalue(\qtreel,Q)=\qtvalue(Q)$ and 
$\childat((\qtreel,Q),i)=\childat(Q,i)$, where we enter the \qtree and
resort to the algorithms based on the 
compact representation of $Q$.
We will assume that $\qtvalue(Q)$ returns \internal for internal nodes, and thus the implementation of $\qtvalue$ for $\extendl$ is trivial (compare Algorithms~\ref{alg:value2} and~\ref{alg:valueextend-body} under this assumption).

\begin{figure}
\begin{minipage}{.45\textwidth}
\begin{algorithm}[H]
\footnotesize
	\begin{algorithmic}[1]
		\Require A \Qtree $Q$.
		\vskip 4pt
		\Ensure Value of the root of $(\notl,Q)$.
		\vskip 10pt
		\State \Return $1-\qtvalue(Q)$
	\end{algorithmic}
	\caption{\qtvalue function for \notl}
	\label{alg:valuenot-body}	
\end{algorithm} 
\end{minipage}
\hskip .03\textwidth
\begin{minipage}{.52\textwidth}
\begin{algorithm}[H]
\footnotesize
	\begin{algorithmic}[1]
		\Require A \Qtree $Q$ in dimension $d$, and an integer $0 \le i < 2^d$.
		\Ensure An \lqdag for the $i$-th child of $(\notl,Q)$.
		\vskip 4pt
		\State \Return $(\notl,\childat(Q,i))$
	\end{algorithmic}
	\caption{\childat function for \notl}
	\label{alg:childatnot-body}	
\end{algorithm}
\end{minipage}
\end{figure}

\begin{figure}
\vskip -14pt
\begin{minipage}{.45\textwidth}
\begin{algorithm}[H]
	\footnotesize
	\begin{algorithmic}[1]
		\Require \Lqdags $L_1$ and $L_2$.
		\Ensure The value of the root of $(\andl,L_1,L_2)$.
		
		\vskip 4pt
		\If {$\qtvalue(L_1)=0$ {\bf or} $\qtvalue(L_2)=0$} \Return $0$ \EndIf
		\If {$\qtvalue(L_1)=1$} \Return $\qtvalue(L_2)$ \EndIf
		\If {$\qtvalue(L_2)=1$} \Return $\qtvalue(L_1)$ \EndIf
		\State \Return \unknown 
	\end{algorithmic}
	\caption{\qtvalue function for \andl}
	\label{alg:valueand-body}	
\end{algorithm}
\end{minipage}
\hskip .03\textwidth
\begin{minipage}{.52\textwidth}
\begin{algorithm}[H]
\footnotesize
	\begin{algorithmic}[1]
		\Require \Lqdags $L_1$ and $L_2$ in dimension $d$, integer
$0 \le i < 2^d$.
                \Ensure An \lqdag for the $i$-th child of $(\andl,L_1,L_2)$.
		\vskip 4pt
		\If {$\qtvalue(L_1)=1$} \Return $\childat(L_2,i)$ \EndIf
		\If {$\qtvalue(L_2)=1$} \Return $\childat(L_1,i)$ \EndIf

		\vskip 4pt
		\State \Return $(\andl,\childat(L_1,i),\childat(L_2,i))$
		\vskip 6pt
	\end{algorithmic}
	\caption{\childat function for \andl}
	\label{alg:childatand-body}	
\end{algorithm}
\end{minipage}
\end{figure}

\begin{figure}[t]
\begin{minipage}{.45\textwidth}
\begin{algorithm}[H]
\footnotesize
	\begin{algorithmic}[1]
		\Require \Lqdags $L_1$ and $L_2$.
		\Ensure The value of the root of $(\orl,L_1,L_2)$.
		\vskip 4pt
		\If {$\qtvalue(L_1)=1$ {\bf or} $\qtvalue(L_2)=1$} \Return $1$ \EndIf
		\If {$\qtvalue(L_1)=0$} \Return $\qtvalue(L_2)$ \EndIf
		\If {$\qtvalue(L_2)=0$} \Return $\qtvalue(L_1)$ \EndIf
		\State \Return \unknown 
	\end{algorithmic}
	\caption{\qtvalue function for \orl}
	\label{alg:valueor-body}	
\end{algorithm}
\end{minipage}
\hskip .03\textwidth
\begin{minipage}{.52\textwidth}
\begin{algorithm}[H]
\footnotesize
	\begin{algorithmic}[1]
		\Require	
		\Lqdags $L_1$ and $L_2$ in dimension $d$, integer $0 \le i < 2^d$.
		\Ensure An \lqdag for the $i$-th child of $(\orl,L_1,L_2)$.
		\vskip 4pt
		\If {$\qtvalue(L_1)=0$} \Return $\childat(L_2,i)$ \EndIf
		\If {$\qtvalue(L_2)=0$} \Return $\childat(L_1,i)$ \EndIf
		
		\vskip 4pt
		\State \Return $(\orl,\childat(L_1,i),\childat(L_2,i))$
		\vskip 6pt
	\end{algorithmic}
	\caption{\childat function for \orl}
	\label{alg:childator-body}	
\end{algorithm}
\end{minipage}
\end{figure}

\begin{figure}[t]
\vspace*{-18pt}
\begin{minipage}{.45\textwidth}
\begin{algorithm}[H]
\footnotesize
        \begin{algorithmic}[1]
			\vskip 5pt	
			\Require
				\Statex \Lqdag $L_1(\A')$, set $\A \supseteq \A'$.
				
				\Ensure
				\Statex Value of the root of $(\extendl,L_1,\A)$.

                \vskip 10pt
				\State \Return $\qtvalue(L_1)$

				\vskip 22pt
        \end{algorithmic}
        \caption{\qtvalue function for \extendl}
        \label{alg:valueextend-body}
\end{algorithm}
\end{minipage}
\hskip .03\textwidth
\begin{minipage}{.52\textwidth}
\begin{algorithm}[H]
\footnotesize
        \begin{algorithmic}[1]
				\Require
				\Statex \Lqdag $L_1(\A')$, set $\A \supseteq \A'$, integer $0 \le i < 2^{|\A|}$.

                \Ensure An \lqdag for the $i$-th child of $(\extendl,L_1,\A)$.

                \vskip 4pt
                \State $d \gets |\A|$, $d' \gets |\A'|$
                \State $m_d \leftarrow $ the $d$-bits binary representation of
$i$
                \LongState{$m_{d'} \leftarrow$ the projection of $m_d$ to the
positions in which the attributes of $\A'$ appear in $\A$}
                \State $i' \leftarrow $ the value in $[0,2^{d'}-1]$
corresponding to $m_{d'}$
                \State \Return $(\extendl,\childat(L_1,i'),\A)$
        \end{algorithmic}
        \caption{\childat function for \extendl}
        \label{alg:childatextend-body}
\end{algorithm}
\end{minipage}
\end{figure}

Note that the recursive calls of
Algorithms~\ref{alg:valuenot-body}-\ref{alg:childatextend-body} traverse the
nodes of the relational formula (fnodes, for short), and terminate immediately 
upon reaching an fnode of the form $(\qtreel,Q)$. Therefore, their time 
complexity depends only on the size of the formula represented by the \lqdag.
We show next how, using these implementations of \qtvalue and \childat, one can efficiently evaluate a relational formula using \lqdags.

\subsection{Evaluating JUC queries}

To evaluate a formula $F$ represented as an \lqdag $L_F$, we compute the 
completion $Q_F$ of $L_F$, that is, the \qtree $Q_F$ representing the 
output of $F$.

To implement this we introduce the idea of \emph{super-completion} of an \lqdag.
The super-completion $Q_F^+$ of $L_F$ is the \qtree induced by navigating
$L_F$, and interpreting the values \unknown as \internal 
(see Algorithm~\ref{alg:supercompletion}).
Note that, by interpreting values \unknown as \internal, we are disregarding
the possibility of pruning resulting subgrids full of $0$s or $1$s and replacing them by single leaves with values $0$ or $1$ in $Q_F$. Therefore, $Q_F^+$ is a 
\emph{non-pruned} \qtree (just as $Q^+$ in Section~\ref{sec:joinalg}) that
nevertheless represents the same points of $Q_F$.
Moreover, $Q_F^+$ shares with $Q_F$ a key property: all its nodes with value $1$, including the
last-level leaves representing individual cells, correspond to actual tuples 
in the output of $F$.

To see how \lqdags are evaluated, let us consider the 
query $F=R(A,B) \bowtie S(B,C) \bowtie \overline{T}(A,C)$. This corresponds to 
an \lqdag $Q_F$:
\begin{eqnarray*}
(\andl,(\andl,(\extendl,(\qtreel,Q_R),\A),(\extendl,(\qtreel,Q_S),\A)),
(\extendl,(\notl,Q_T),\A)).
\end{eqnarray*}

Assuming some of the trees involved have internal nodes, the super-completion
$Q_F^+$ first produces $8$ children. 
Suppose the grid of $T$ is full of $1$s in the first quadrant ($00$).
Then the first child ($00$) of $Q_T$ has value $1$, which becomes value $0$ in
$(\notl,Q_T)$. This implies that $(\extendl,(\notl,Q_T))$ also yields value $0$ in octants $000$ and $010$.
Thus, when function $\childat$ is called on child $000$ of $Q_F$, our $0$ is 
immediately propagated and \childat returns $0$, meaning that there are no 
answers for $F$ on this octant, without ever consulting the \qtrees $Q_R$
and $Q_S$ (see Figure~\ref{fig:example3} for an illustration).
On the other hand, if the value of the child $11$ of $T$ is $0$, then
$(\extendl,(\notl,Q_T))$ will return value $1$ in octants $101$ and 
$111$. This means that the result on this octant corresponds to the result of
joining $R$ and $S$; indeed \childat towards $101$ in $Q_F$ returns
\begin{eqnarray*}
(\andl,\childat((\extendl,(\qtreel,Q_R),\A),101), \childat((\extendl,(\qtreel,Q_S),\A),101)).
\end{eqnarray*}
If $\childat((\extendl,(\qtreel,Q_R),\A),101)$ and 
$\childat((\extendl,(\qtreel,Q_S),\A),101)$ are trees with internal nodes, 
the resulting \andl can be either an internal node or a leaf with value $0$
(if the intersection is empty), though not a leaf with value $1$. Thus, for 
now, the \qtvalue of this node is unknown, a \unknown.
See Figure~\ref{fig:example3} for an illustration.

\begin{SCfigure}[.78][t]
	\includegraphics[page=5,scale=.8]{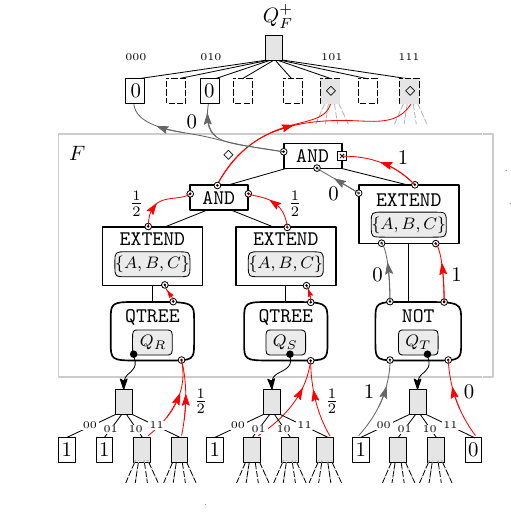}
	\vspace*{-.1cm}	
	\caption{
		Illustration of the syntax tree of an \lqdag 
		for the formula $(R(A,B) \bowtie S(B,C)) \bowtie \overline{T}(A,C)$.
		The \qtrees $Q_R, Q_S, Q_T$ represent the relations $R, S, T$, respectively. We show the top values of $Q_F^+$ on top and of $Q_T$ on the bottom. The gray upward arrows show how the value $1$ in the quadrant $00$ of $Q_T$ becomes $0$s in octants $000$ and $010$ of $Q_F^+$ without accessing $Q_R$ or $Q_S$. The red upward arrows show how the value $0$ in the quadrant $11$ of $Q_T$ makes the quadrants $101$ and $111$ of $Q_F^+$ depend only on their left child 
		(and, assuming their value is \internal, becomes a \unknown in $Q_F^+$).\\
		~
	}
	\label{fig:example3}
\end{SCfigure}

Note that the running time of Algorithm~\ref{alg:supercompletion} is 
$O(|Q_F^+|)$. 
One can then compact $Q_F^+$ to obtain $Q_F$, in time $O(|Q_F^+|)$ as well,
with a simple bottom-up traversal.
Thus, bounding $|Q_F^+|$ yields a bound for the running time of evaluating $F$.
While $|Q_F^+|$ can be considerably larger than the actual size $|Q_F|$ of the
output, we show that $|Q_F^+|$ is bounded by the worst-case output size of 
formula $F$ for a database with relations of approximately the same size. 
To prove this, the introduction of values \unknown plays a key role.%
\footnote{In an implementation, we could simply use \internal instead of 
\unknown, without indicating that we are not yet sure that 
the value is \internal: we build $Q_F^+$ assuming it is, and only make sure 
later, when we compact it into $Q_F$.}

\begin{algorithm}[t]
\footnotesize
	\begin{algorithmic}[1]
		\Require An $\lqdag$ $L_F$ whose completion represents a formula $F$ over relations with $d$ attributes.
		\Ensure The super-completion $Q_F^+$ of $L_F$.
		
		\vskip 4pt
		\If {$\qtvalue(L_F) \in \{0,1\}$} \Return a leaf with value 
			$\qtvalue(L_F)$ \EndIf
		\LongState{\Return an internal node with children\\ 
			{\color{white} . \;\;}
			\Big($\texttt{\sc SCompletion}\big(\childat(L_F, 0)\big), \ldots,$		
			$\texttt{\sc SCompletion}\big(\childat(L_F, 2^d-1)\big)$\Big)}
		
	\end{algorithmic}
	\caption{\texttt{\texttt{\sc SCompletion}}}
	\label{alg:supercompletion}
\end{algorithm}

\smallskip
\noindent
\textbf{The power of the \unknown values}. 
Consider again Algorithm \ref{alg:supercompletion}. The lowest places in $L_F$ where \unknown values are introduced are the 
\qtvalue of \andl and \orl \lqdags where both operands have $\qtvalue=$
\internal. We must then set the $\qtvalue$ to \unknown instead of \internal
because, depending on the evaluation of the children of the 
operands, the \qtvalue can turn out to be actually $0$ for \andl or $1$ for \orl.
Once produced, a value \unknown is inherited by the ancestors in the formula
unless the other value is $0$ (for \andl) or $1$ (for \orl).

Imagine that a formula $F$ involves $n$ relations $R_1, \ldots, R_n$
represented as \qtrees in dimension $d$, including no negations.
Suppose that we {\em trim} the \qtrees of $R_1, \ldots, R_n$ by removing all
the levels at depth higher than some $j$ (thus making the $j$-th level the last 
one) and assuming that the internal nodes at level $j$ become
leaves with value $1$. We do not attempt to compact the resulting \qtrees, so
their nodes at levels up to $j-1$ stay identical and with the same \qtvalue. 
If we now compute $Q_F^+$ over those (possibly non-pruned) \qtrees, the 
computation will be identical up to level $j-1$, and in level $j$ every 
internal node in the original $Q_F^+$, which had value \internal, will now
operate over all $1$s, and thus will evaluate to $1$ because \andand and \oror
are monotonic.

Thus, these nodes belong to the output of $F$ over the relations $R_1',
\ldots,R_n'$ induced by the trimmed \qtrees (on smaller domains of size $\ell'
= 2^j$), 
with sizes $|R_1|' \le |R_1|, \ldots,\! |R_n|' \le |R_n|$.
This would imply, just as in the proof of Theorem~\ref{thm:main}, a bound on
the maximum number of nodes in a level of $Q_F^+$, thus proving the worst-case
optimality of the size of $Q_F^+$ (up to $\log min(\ell, N)$-factors, and factors depending on the query size), and thus the worst-case
optimality of Algorithm~\ref{alg:supercompletion} in data complexity.

However, this reasoning fails when one trims at the $j$-th level a \qtree $Q$
that appears in an \lqdag $L=(\notl, Q)$, because the value $1$ of the nodes
at the $j$-th level of $Q$ after the trimming changes to $0$ in $L$.
So, to prove that our algorithm is worst-case optimal we cannot rely only on relations obtained by trimming those that appear in the formula. We need
to generate new \qtrees for those relations under a \notl operation that preserve the values of the completion of \notl after the trimming.
Next we formalize how to do this.

\smallskip
\noindent
\textbf{Analysis of the algorithm}.
Let $L_F$ be an \lqdag for a formula $F$.
The \emph{syntax tree} of $F$ is the directed tree formed by the fnodes in
$F$, with an edge from fnode $L$ to fnode $L'$ if $L'$ is an operand of $L$.
The leaves of this tree are always {\em atomic expressions}, that is, the
fnodes, with functors \qtreel and \notl, that operate on one \qtree (see Figure~\ref{fig:example3} again).
We say that two atomic expressions $L_1$ and $L_2$ are equal if both their
functors and operands are equal.
For example, in the formula 
\[	F = (\orl, (\andl, (\qtreel, Q_R), (\qtreel, Q_S)), 
		   (\andl, (\qtreel, Q_R), (\qtreel, Q_T)))
\]
there are three different atomic expressions, $(\qtreel,Q_R)$, $(\qtreel,Q_S)$,
and $(\qtreel, Q_T)$, while in $F'=(\andl, (\qtreel, Q_R), 
(\notl, Q_R))$ there are two atomic expressions.
Notice that in formulas like $F'$, where a relation appears both negated and
not negated, the two occurrences are seen as different atomic expressions. We return later to the consequences of this definition.

The following lemma is key to bound the running time of Algorithm~\ref{alg:supercompletion} while evaluating a formula $F$.

\begin{restatable}{lemma}{wcogeneral}
\label{lem:wco-general}
Let $F$ be a relational formula represented by an \lqdag $L_F$ in
dimension $d$, and let $Q_F^+$ be the super-completion of $F$.
	Let $Q_1, \ldots, Q_n$ be the \qtree operands of the different atomic expressions of $F$, and $R_{1}(\A_1), \ldots, R_{n}(\A_n)$ be the (not necessarily different) relations represented by these \qtrees, respectively.
	Let $M$ be the maximum number of nodes in a level of $Q_F^+$.
	Then, there is a database with relations $R_1'(\A_1), \ldots,
R_n'(\A_n)$ of respective sizes $O(2^d|Q_1|), \ldots, O(2^d|Q_n|)$, such that
the output of $F$ evaluated over it 
has size
$\Omega(M/2^d)$.
\end{restatable}
\begin{proof}
Let $m_l$ be the number of nodes in level $l$ of $Q_F^+$ and $j$ be a level
where $M = m_j$ is maximum. 
	We assume that $j>1$, otherwise $M=O(1)$ and the result is trivial.
	We first bound the number of nodes with value \internal at the
$(j-1)$-th level. 
	By hypothesis, $m_j \ge m_{j-1}$, and since a node in $Q_F^+$ is
present at level $j$ only if its parent at level $j-1$ has value \internal, in
the $(j-1)$-th level there are at least $m_j/2^d$ nodes with value
\internal. 
	
	Now, let $A_1, \ldots, A_n$ be the atomic expressions of $F$,
	and let $Q_1',\ldots, Q_n'$ be the \qtrees that result from trimming the levels at depths higher than $j-1$ from $Q_1, \ldots, Q_n$, respectively.
	Consider the completion $A_i^*$ of $A_i$ evaluated over $Q_i$, and the
completion ${A_i^{*}}'$ of $A_i$ evaluated over (the possibly non-pruned) $Q_i'$, for all $1 \le i \le n$.
	If it is always the case that the first $j-1$ levels of $A_i^*$ are respectively equal to the $j-1$ levels of ${A_i^{*}}'$ then we are done.
	To see why, let ${Q_F^+}'$ be the super-completion of $F$ when evaluated over $Q_1',\ldots, Q_n'$.
	The first $j-2$ levels of $Q_F^+$ will be the same as those of ${Q_F^+}'$ because the same results of the operations are propagated up from the leaves of the syntax tree of $F$ before and after the trimming.
	Moreover, in the $(j-1)$-th level ${Q_F^+}'$ (its last level) the nodes
with value $1$ are precisely the nodes with value $1$ or \internal in
${Q_F^+}$, where we note that: (i) there are at least $m_j/2^d$
of them; and (2) they belong to the output of $F$ over the relations $R_1', \ldots, R_n'$ represented by $Q_1',\ldots, Q_n'$. 

	We know that $|R_1'| \le |R_1|, \ldots, |R_n'| \le |R_n|$.
	However, the values of $R_1', \ldots, R_n'$ correspond to a smaller universe. 
	This can be remedied by simply appending $(\log \ell-j)$ 0's at the beginning of the binary representation of these values.
	This would yield the desired result: we have $n$ relations
over the same set of attributes as the original ones, with same respective
cardinality, and such that when $F$ is evaluated over them the output size is
$\Omega(m_j/2^d)$.
	
	However, for atomic expressions of the type $A_i=(\notl,Q_i)$ it is not the case that the first $j-1$ levels of ${A_i^{*}}$ coincide with the $j-1$ levels of ${A_i^{*}}'$.
	Anyway, we can deal with this case: their first $j-2$ levels will coincide, and in the last level, the value of a node present in ${A_i^{*}}$ is the negation of the value of the homologous node in ${A_i^{*}}'$.
	Thus, instead of choosing the \qtree $Q_i'$ that results from trimming
$Q_i$, we choose the \qtree $Q_i''$ in which the first $j-2$ levels are the
same as $Q_i'$, and the $(j-1)$-th level results from negating the value of every node in $Q_i'$.
	Note that if we let now  ${A_i^{*}}'$ be the completion of $A_i$ evaluated over $Q_i''$, then the first $j-1$ levels of ${A_i^{*}}$ will be exactly same as the $j-1$ levels of ${A_i^{*}}'$. Finally, note that the size of the relation represented by $Q_i''$ cannot be larger than $2^d|Q_i|$.
	The result of the lemma follows.
\end{proof}

Using the same reasoning as before we can  bound the time needed to compute 
the super-completion $Q_F^+$ of an \lqdag $L_F$ in dimension $d$ involving
\qtrees representing $R_{1}, \ldots, R_{n}$. Since $M$ is the 
maximum number of nodes in a level of $Q_F^+$, the number of nodes in $Q_F^+$ is 
at most $M \log \ell$. Now, each node in $Q_F^+$ results from the application 
of $|F|$ operations on each of the $2^d$ children being generated, all of which
take constant time. Thus the super-completion can be computed in time 
$O(M \cdot 2^d |F| \log \ell)$. If we use $F(D)^*$ to denote the size of the 
maximum output of the query $F$ over instances with relations $R_1, 
\ldots, R_n$ of respective sizes $O(2^d|Q_1|), \ldots, O(2^d|Q_n|)$, 
then by Lemma~\ref{lem:wco-general} the query $F$ can be computed in time 
$O(F(D)^* \cdot 2^{2d} |F| \log \ell)$. This means that the algorithm is indeed
worst-case optimal. 

The requirement of different atomic expressions 
is because we need to consider $R$ and $\notnot\ R$ as different relations. To see this, consider again our example formula $F'=(\andl, (\qtreel, Q_R), 
(\notl, Q_R))$. We clearly have that the answer of this query is always empty, and therefore $|Q_{F'}|=0$. However, here $|Q^+_{F'}|=\Theta(|R|)$ for every $R$, and thus our algorithm is
worst-case optimal only if we consider the possible output size of a
more general formula, $F'' = (\andl, (\qtreel, Q_R), (\notl, Q_R'))$.
This impacts in other operations of the relational algebra. We can write all of them as \lqdags, but for some of them we will not ensure their optimal evaluation. For instance, the expression 
$Q_R\ \andand\ (\notnot\ (Q_R\ \andand\ Q_S))$, which expresses the antijoin between $R$ and $S$, is not optimal because both $Q_R$ and $\notnot\ Q_R$ appear in the formula. 
A way to ensure that our result applies is to require that the atomic
expressions (once the \notl operations are pushed down) refer all to different 
relations.

\begin{theorem}
	Let $F$ be a relational formula represented by an \lqdag $L_F$.
	If the number of different relations involved in $F$ equals the number of different atomic expressions, then Algorithm~\ref{alg:supercompletion} evaluates $F$ in worst-case optimal time in data complexity.
\end{theorem}

Note that this result generalizes Theorem~\ref{thm:main}. 
Moreover, it does not
matter how we write our formula $F$ to achieve worst-case optimal evaluation.
For example, our algorithms behave identically on $((R \bowtie S) \bowtie T)$
and on $(R \bowtie (S \bowtie T))$.

\section{Final Remarks}\label{sec-discuss}
The evaluation of join queries using \qdags provides a competitive alternative to
current worst-case optimal
algorithms~\cite{geometric,panda,minesweeper,NPRR12,leapfrog}. When compared to
them, we find the following time-space tradeoffs.

Regarding space, \qdags require only just a few extra words per tuple, which is
generally much less than what standard database indexes require, and definitely
less than those required by current worst-case optimal algorithms (e.g.
\cite{geometric,NPRR12,leapfrog}).
Moreover, in both NPRR~\cite{NPRR12} and leapfrog~\cite{leapfrog}, the required
index structure only works for a specific ordering of the attributes. Thus, in
order to efficiently evaluate any possible query using these two
algorithms, a separate index is required for every possible attribute order (i.e.,
$d!$ indexes). In contrast, all we need to store is one \qtree per relation, and
that works for any query.
Even if we resort to the (simpler) $k^d$-tree representation by Brisaboa et
al.~\cite{BLNis13}, the extra space increases by a factor of $2^d$ bits, which
is still considerably less than the alternative of $d!$ standard indexes for any
order of variables (e.g., for $d=10$, $2^d = 1024$, while $d! =
3628800$, i.e., $\approx 3500$ times bigger).

Regarding time, the first comparison that stands aside is the $\log(N)$ factor,
present in our solution but not in others like NPRR~\cite{NPRR12} and
leapfrog~\cite{leapfrog}. Note, however, that NPRR assumes to be able to compute
a join of two relations $R$ and $S$ in time $O(|R| + |S| + |R \bowtie S|)$,
which is only possible when using a hash table and when time is computed in an
amortized way or in expectation \cite[footnote 3]{NPRR12}. This was
also noted for leapfrog~\cite[Section 5]{leapfrog}, where they state that
their own $\log(N)$ factor can be avoided by using hashes instead of tries, but
they leave open whether this is actually better in practice.
More involved algorithms such as PANDA~\cite{panda} build upon algorithms to
compute joins of two relations, and therefore the same $\log(N)$ factor appears
if one avoids hashes or amortized running time bounds.
Our algorithm incurs in an additional $2^d$ factor in time when compared to NPRR
or leapfrog, similarly to other worst-case optimal solutions based on
geometric data structures~\cite{geometric,minesweeper}. This is, as far as we
are aware, unavoidable: it is the price to pay for using so little space. Note,
however, that this factor does not depend on the data, and that it can be
compensated by the fact that our native indexes are compressed, and thus might fit
entirely in faster memory.

One important benefit of our framework is that answers to queries can be
delivered in their compact representation. As such, we can iterate over them, or
store them, or use them as materialized views, either built eagerly, as \qtrees,
or in lazy form, as \lqdags. One could even cache the top half of the
(uncompacted) tree containing the answer, and leave the bottom half in the form
of \lqdags. The upper half, which is used the most, is cached, and  the bottom
half is computed on demand. Our framework also permits sharing \lqdags as common
subexpressions that are computed only once.
Additionally, we envision two main uses for the techniques presented in this
paper. On one hand, one could take advantage of the low storage cost of these
indexes, and add them as a companion to a more traditional database setting.
Smaller joins and selections could be handled by the database, while multijoins
could be processed faster because they would be computed over the \qtrees. On the other hand, we could use \lqdags\ instead, so as to evaluate more expressive
queries over \qtrees. Even if some operations are not optimal, what is lost in
optimality may be gained again because these data structures allow operating in
faster memory levels. 

There are several directions for future work. For instance, we are trying to
improve our structures to achieve good bounds for acyclic queries (see Appendix
\ref{sec-more-stuff}), and we see an opportunity to apply quadtrees in the
setting of parallel computation (see, e.g., Suciu \cite{parallel}).
We also comment in  Appendix~\ref{sec-more-stuff} on bounds for clustered databases,
another topic deserving further study. 

%
\bibliographystyle{plainurl}
\bibliography{joins_trees}

\appendix
\section{Appendix}
\label{sec-more-stuff}

\subsection{Additional comments on Projections}
Including projection in our framework is not difficult: in a quadtree $Q$ storing a relation $R$ with attributes $\A$, one can compute the projection $\pi_{\A'}(R)$, for $\A' \subseteq \A$ as follows. 
Assume that $|\A| = d$ and $|\A'| = d'$. Then the projection is the quadtree defined inductively as follows. If $\qtvalue(Q)$ is $0$ or $1$ then the projection is a leaf with the same value. 
Otherwise $Q$ has $2^{d}$ children. The quadtree for  $\pi_{\A'}(R)$ has instead $2^{d'}$ children, where the $i$-th child is defined as the \oror of 
all children $j$ of $Q$ such that the projection of the $d$-bit representation of $j$ to the positions in which attributes in $\A'$ appear in $\A$ is precisely the 
$d'$-bit representation of $i$. For example, computing $\pi_{A_1,A_2}R(A_1,A_2,A_3)$ means creating a tree with four children, resulting of the \oror of children $0$ and $1$, $2$ and $3$, $4$ and $5$ and $6$ and $7$, respectively. 

Having defined the projection, a natural question is whether one can use it to obtain finer bounds for acyclic queries or for queries with bounded treewidth. For example, 
even though the AGM bound for $R(A,B) \bowtie S(B,C)$ is quadratic, one can use Yannakakis' algorithm \cite{yannakakis} to compute it in time
$O(|R|+|S| + |R \bowtie S|)$. This is commonly achieved by first computing $\pi_B(R)$ and $\pi_B(S)$, joining them, 
and then using this join to filter out $R$ and $S$. 
Unfortunately, adopting this strategy in our lqdag framework would still give us a quadratic algorithm, even for queries with small output, because after the projection we would need to extend the result again. The same holds for the general Yannakakis' algorithm when computing the final join after performing all necessary semijoins.

More generally, this also rules out the possibility to achieve optimal bounds for queries with bounded treewidth or similar measures. Of course, this is not much of a limitation because 
one can always compute the most complex queries with our compact representation and then carry out Yannakakis' algorithms on top of these results with standard database techniques, but it would be 
better to resolve all within our framework. We are currently looking at improving our data structures in this regard. 

\subsection{Geometric representation and finer analysis}
As \qtrees have a direct geometric interpretation, it is natural to compare them to the algorithm based on \emph{gap boxes} proposed by Khamis et al.~\cite{geometric}. In a nutshell, this algorithm uses a data structure that stores relations as a set of multidimensional cubes that contain no data points, which the authors call gap boxes. Under this framework, a data point is in the answer of the join query $R_1 \bowtie \cdots \bowtie R_n$ if the point is not part of a gap box in any of the relations $R_i$. The authors then compute the answers of these queries using an algorithm that finds and merges appropriate gap boxes covering all cells not in the answer of the query, until no more gap boxes can be found and we are left with a covering that misses exactly those points in the answer of the query. Perhaps more interestingly, the algorithm is subject of a finer analysis: the runtime of queries can be shown to be bounded by a function of the size of a \emph{certificate} of the instance (and not its size). The certificate in their case is simply the minimum amount of gap boxes from the input relations that is needed to cover all the gaps in the answer of the query. Finding such a minimal cover is NP-hard, but a slightly restricted notion of gap boxes maintains the bounds within an $O(\log^d \ell)$ approximation factor. 

While any index structure can be thought of as providing a set of gap boxes \cite{geometric}, \qtrees provide a particularly natural and compact representation. Each node valued $0$ in a \qtree signals that there are no points in its subgrid, and can therefore be understood as a $d$-dimensional gap box. We can understand \qdags as a set of gap boxes as well: precisely those in its completion. Now let $J = R_1 \bowtie \cdots \bowtie R_n$ be a join query over $d$ attributes, and let $R_1^*,\dots,R_n^*$ denote the extension of each $R_i$ with the attributes 
of $J$ that are not in $R_i$. As in Khamis et al.~\cite{geometric}, a \emph{\qtree certificate} for $J$ is a set of gap boxes (i.e., empty $d$-dimensional grids obtained from any 
of the $R_i^*$s) such that every coordinate not in the answer of $J$ is covered by at least one of these boxes. Let $C_{J,D}$ denote a certificate for $J$ of minimum size. 

\begin{proposition}
Given query $J$ and database $D$, Algorithm~\ref{alg:join} runs in time $O((|C_{J,D}| + |J(D)|) \cdot 2^d n \log \ell)$, where $J(D)$ is the output of the query $J$ over $D$. 
\end{proposition}

Now, one can easily construct instances and queries such that the minimal certificate $C_{J,D}$ is comparable to $2^{\rho^*(J,D)}$. So  this will not give us 
optimality results, as discovered  \cite{geometric,minesweeper} for acyclic queries or queries with bounded treewidth. This is a consequence of increasing the dimensionality of the relations.  
Nevertheless, the bound does yield a good running time when we know that $C_{J,D}$ is small. It is also worth mentioning that our algorithms directly computes the only possible representation of the output as gap boxes (because its boxes come directly from the representation of the relations). This means that there is a direct connection between instances that give small certificates and instances for which the representation of the output is small. 

\subsection{Better runtime on clustered databases} 

\Qtrees have been shown to work well in applications such as RDF stores or web graphs, where 
data points are distributed in clusters \cite{BLNis13,AGBFMPNkais14}. It turns out that combining the analysis described in Section~\ref{sec-qtrees} for clustered grids with the technique we used to show 
 that joins are worst-case optimal, results in a better bound for the running time of our algorithms, and a small refinement of the AGM bound itself. 

 Consider again the triangle query $R(A,B) \bowtie S(B,C) \bowtie T(A,C)$, and assume the points in each relation are distributed in $c$ clusters, each of them 
 of size at most $s \times s$, and with $p$ points in total. Then, at depth $\log (\ell/s)$, the \qtrees of $T$, $R$, and $S$ have at most $2^d$ internal nodes per cluster (where we are in dimension $d=3$): at this level 
 one can think of the trimmed quadtree as representing a coarser grid of cells of size $s^d$, and therefore each cluster can intersect at most two of these coarser 
 cells per dimension. Thus, letting $Q_R'$, $Q_S'$, and $Q_T'$ be the \qtrees for $R$, $S$ and $T$ trimmed up to level $\log (\ell/s)$ 
 (and where internal nodes take value $1$), then the proof of Theorem~\ref{thm:main} yields a bound for the number of internal nodes at level $\log (\ell/s)$ of the 
 \qtree $Q^+$ of the output before the compaction step (or, equivalently, of the super-completion of the \lqdag of the triangle query): this number must be bounded 
 by the AGM bound of the instances given by $Q_R'$, $Q_S'$ and $Q_T'$, which is at most $(c \cdot 2^d)^{3/2}$. Going back to the data for the quadtree $Q^+$, the 
 bound on the number of internal nodes means that the points of the output are distributed in at most $(c \cdot 2^d)^{3/2}$ clusters of size at most $s^d$. 
 In turn, the maximal number of $1$s in the answer is bounded by the AGM bound itself, which here is $p^{3/2}$. This means that the size of $Q^+$ is bounded by  
 $O((c \cdot 2^d)^{3/2} \log \ell + p^{3/2} \log s)$, and therefore the running time of the algorithm is $O\big(((c \cdot 2^d)^{3/2} \log  \ell + p^{3/2} \log s) \cdot 2^d\big)$. 
 This is an important reduction in running time if the number $c$ of clusters and their width $s$ are small, as we now multiply the number of answers by $\log s$ instead of $\log \ell$. 

 To generalize, let us use $D^{c,d}$ as the database ``trimmed'' to $c \cdot 2^d$ points. The discussion above can be extended to prove the following.  

 \begin{proposition} \label{prop:clusters}
 Let $J = R_1 \bowtie \cdots \bowtie R_n$ be a full join query, and $D$ a database over schema $\{R_1,\ldots,R_n\}$, with $d$ attributes in total, where the 
 domains of the relations are in $[0,\ell-1]$, and where the points in each relation are distributed in $c$ clusters of width $s$. 
 Then Algorithm~\ref{alg:join} works in time $O\big(  (2^{\rho^*(J,D^{c,d})} \log \ell + 2^{\rho^*(J,D)} \log s) \cdot 2^d n )$.
 \end{proposition}

\end{document}